\documentclass[11pt]{article}
\usepackage{amssymb,amsmath,amsthm,url}
\usepackage{graphicx}
\usepackage{labels}
\usepackage{algorithmic}
\usepackage{algorithm}
\usepackage{verbatim}
\usepackage{hyperref}

\usepackage{tikz}
\usetikzlibrary{arrows}
\usetikzlibrary{arrows,%
                shapes,positioning}
\usetikzlibrary{automata,matrix,backgrounds,calc}
\usetikzlibrary{decorations.text}
\usetikzlibrary{decorations.pathmorphing}
\usetikzlibrary{shapes.geometric}
\usepackage[utf8]{inputenc}                
\usetikzlibrary{trees}

\usepackage{pstricks}
\usepackage[tiling]{pst-fill}

\usepackage{mflogo}
 \input{random.tex} 
 \usepackage{url} 
 \usepackage{multido} 

\topmargin by -\headsep

\newtheorem{theorem}{Theorem}

\newtheorem{lemma}[theorem]{Lemma}

\newtheorem{corollary}{Corollary}

\theoremstyle{definition}

\newtheorem{definition}{Definition}

\newtheorem{example}{Example}
\newtheorem{problem}{Problem}

\title{\bf An $O(n^{0.4732})$ upper bound on the complexity of the GKS communication game}

\author{
        \textsc{Mario Szegedy}\\ %
        Department of Computer Science\\
        Rutgers, The State University of New Jersey\\
        \ %
        \normalsize
            \texttt{szegedy@cs.rutgers.edu}
       }

\begin{document}

\maketitle

\begin{abstract}
We give an $5\cdot n^{\log_{30}5}$ upper bund on the complexity of the communication game introduced by G. Gilmer, M. Kouck\'y and M. Saks \cite{saks} to study the Sensitivity Conjecture \cite{linial}, improving on their 
$\sqrt{999\over 1000}\sqrt{n}$ bound. We also determine the exact complexity of the game up to $n\le 9$.
\end{abstract}
    
\section{The $O(n^{0.4732})$ upper bound}

{\bf The GKS communication game}, defined by G. Gilmer, M. Kouck\'y and M. Saks \cite{saks} is played by two cooperating players, Alice and Bob, against an all-powerful adversary, Merlin. 
The game has a single parameter $n$. Merlin has a permutation $\pi = \pi_{1}\pi_{2}\ldots \pi_{n}$ of $[n]$ and a bit $b$. Alice has a strategy 
$S:\{\mbox{partial permutations on $[n]$}\}\rightarrow \{0,1\}$ and Bob has a strategy
$T: \{0,1\}^{n} \rightarrow 2^{[n]}$. 

\medskip

In Phase 1. Alice assigns zeroes and ones to all but one entries of an array $A[1..n]$ and Merlin sets the remaining entry to 0 or 1
according to the {\em Alice-Merlin protocol} described below. In Phase 2 Bob 
has to guess which entry was set by Merlin by merely looking at $A_{\rm final}$, where $A_{\rm final}$ is the setting of $A$ when Phase 1 is finished.
Bob's guess, $T(A_{\rm final}) \subseteq [n]$, is a subset of entries of $A$ that must include the entry Merlin has set. This has to hold for every strategy $\pi,b$ of Merlin.

\medskip

{\bf Alice-Merlin protocol:} For $1\le i \le n-1$ Alice sets 
\[
A[\pi_{i}] = S(\pi_{1}\ldots\pi_{i}) 
\]
In the end Merlin sets $A[\pi_{n}] = b$.

\begin{definition}
A $(k,n)$ strategy for the GKS game with parameter $n$ is a pair $S,T$ as above such that in addition $|T(\sigma)|\le k$ for every $\sigma\in\{0,1\}^n$. 
\[
k(n) = \min_{k} \mbox{ There is a $(k,n)$ strategy for the GKS game with parameter $n$}
\]
\end{definition}

{\bf The relevance} of the GKS communication game is that $k(n)$ gives a lower bound on the
sensitivity, $s(f) = \max_{x} |\{i\mid f(x\oplus e_{i})\neq f(x)\}|$ for any Boolean function $f:\{0,1\}^{n}\rightarrow \{0,1\}$ with Fourier degree $\deg(f)=n$. See \cite{saks}. See \cite{NS, harry, HKP} for more background.
In turn, any lower bound $k(n) \in \Omega(n^{\alpha})$ for some $\alpha >0$ would positively resolve the 
long-standing Sensitivity Conjecture \cite{linial} which says that the sensitivity and the Fourier degree are polynomially related.

\begin{example}
There exists a $(2,4)$ strategy for the GKS game as follows: Alice (mentally) decomposes $\{1,2,3,4\}$ into blocks,  $\{1,2\}$ and $\{3,4\}$.
When Merlin gives a position in a yet untouched block, Alice answers with 0, and the second time a block is touched she answers with 1 (unless it is the last entry overall and so Merlin's turn).
Assume now that what Bob sees is $A=[1,0,0,1]$. Then he knows that the last bit had to be at position 1 or at 4 (at 2 or 3 it could not be, since in that case
the first position touched in that block is 1, contrary to Alice's strategy). If $A=[0,0,1,0]$ or $A=[1,1,1,0]$, the last bit had to be at position either 1 or 2, since Bob can deduce that 
Merlin had to interfere in the 
first block (in both final results) for the situation to arise. 
These are the only possible cases up to a permutation that respects the blocks.
\end{example}

\begin{lemma}\label{product}
If $(k,n)$ and $(k',n')$ strategies exist, then there is also a $(kk',nn')$ strategy.
\end{lemma}

\begin{proof} Let $S,T$ be a $(k,n)$ strategy and $S',T'$ be a $(k',n')$ strategy. We design a $(kk',nn')$ strategy as follows: Decompose the $nn'$ elements into 
$n'$ blocks of size $n$, e.g. in the fashion 
\[
\{1,2,3\},\{4,5,6\},\{7,8,9\},\{10,11,12\} \;\;\;\;\; (n=3,n'=4)
\]  

\medskip

{\bf Alice's strategy:}

\medskip

\begin{enumerate}
\item Until the last element of any given block is reached, Alice follows strategy $S$ restricted to that block. She does this independently for all blocks.   
\item  When the last element of a block $i$ is reached, Alice decides at the value of the corresponding bit in such a way that the {\em sum of the bits modulo two} in the block 
agrees with the bit that strategy $S'$ would give to the single entry, $i$, in a corresponding situation. More precisely:

Alice mentally runs strategy $S'$ on an array $A'[1..n']$ with indices corresponding to the blocks of the compound game. 
Every time when a block $i$ in the compound game is about to be completed, she computes the entry $i$ of $A'$, by strategy $S'$. 
At the same time she gives an assignment to the last entry of block $i$ in the 
compound game in such a way that with the new bit the mod 2 entry-sum of block $i$ equals to $A'[i]$. However, if block $i$ is the very last one
to be completed, it is Merlin's turn to assign the last bit. The mod 2 entry-sum of block $i$ now gives some arbitrary evaluation of $A'[i]$, which is fine,
since the last entry of $A'$ is Merlin's move in the $S',T'$ game as well.

Assume for instance that blocks $3,4$ are completed (in this order) and now Alice is to evaluate the single remaining bit of block $2$ to Merlin's order.
Then she finds out the bit that strategy $S'$ gives to $A'[2]$ in the situation when $A'[3]$ and $A'[4]$ were set in this order.
If at this point the assignment (in the composed game) is 
$(**0)(*10)(001)(101)$, and strategy $S'$ says that with the $342$ permutation-start Alice's evaluation of position 2 is $0=S'(342)$, Alice  
must evaluate the last entry of block 2 as $1$, because this makes the mod 2 sum of 
the second block $0=S'(342)$. Thus the new assignment is $(**0)(110)(001)(101)$.
\end{enumerate}

\medskip

{\bf Bob's strategy:}

\medskip

Bob computes the mod 2 sum for each block of the array of the compound game to get an array $A'_{\rm final}$ of length $n'$. Then he computes
$T'(A'_{\rm final})$. This gives him at most $k'$ indices. The crucial last entry of the compound game must come from a block indexed from $T'(A'_{\rm final})$.
By applying $T$ on each of these blocks separately, Bob gets at most $k'k$ entries total, and it is easy to see that they are the only candidates for the last entry.
\end{proof} 

G. Gilmer, M. Kouck\'y and M. Saks have proven the existence of a 
$\left( \left\lceil \sqrt{999\over 1000}\sqrt{n}\right\rceil , n\right)$ strategy  \cite{saks}. We describe a modification of the construction in \cite{saks}, 
which together with our first lemma will give the $5 \cdot n^{\log_{30}5}$ upper bund. 

\begin{lemma}\label{the526lemma}
There is a $(5,30)$ strategy for the GKS game.
\end{lemma}
\begin{proof}
We first describe Alice's strategy. Before the game she mentally decomposes the entries of $A$ into 5 blocks, each of length 6, in the fashion
\[
\{1,2,\ldots,6\},\{7,8,\ldots,12\},\{13,14,\ldots,18\},\{19,20,\ldots,24\},\{25,26,\ldots,30\}
\]
For each block $j$  ($0\le j\le 4$) Alice performs the following (identical) protocol, independently:
When Merlin tells Alice to access the block for the first time, and the first requested element from the block has index $6j + i$,
Alice looks up $w_{i}$ from
\medskip

\begin{center}
\begin{tabular}{llllllll}
$w_{1}$ & = & 0 & 0 & 0 & 0 & 0 & 0 \\
$w_{2}$ & = & 1 & 0 & 0 & 1 & 1 & 0 \\
$w_{3}$ & = & 0 & 1 & 0 & 1 & 0 & 1 \\
$w_{4}$ & = & 0 & 0 & 1 & 0 & 1 & 1 \\
$w_{5}$ & = & 1 & 1 & 1 & 0 & 0 & 0 \\
$w_{6}$ & = & 1 & 1 & 1 & 1 & 1 & 1 \\
\end{tabular}
\end{center}

\medskip

When she evaluates any element with index $6j + i'$, except the last one from the block, she sets $A[6j + i']$ to $w_{i}[i']$. She sets the last entry 
from the block to $1-w_{i}[i']$.

\medskip

Notice that the Hamming distance between any $w_{i'}$, $w_{i''}$ ($i'\neq i''$) is at least three. This gives Bob the following recovery strategy:
Regardless whether Merlin or Alice has set the last bit of block $j$, 
Bob can decode $i=i(j)$ for that block, because the Hamming distance of the block from $w_{i(j)}$ is at most one in both cases. If the last remaining bit of block $j$ was controlled by Alice,
it must be the only $i'$ with 
\begin{equation}\label{bobeq}
A[6j + i'] = 1-w_{i}[i']
\end{equation}
If this equation fails to hold for all $1\le i' \le 6$ for some block $j_{0}$, Bob knows that the last bit of that block was set by Merlin. In this case Bob's output is
the set $\{6j_{0} + i' \mid \; 1\le i' \le 6 \; \wedge \; i'\neq i(j_{0})\}$. He could exclude $i' = i(j_{0})$, since he knows that $6j_{0} + i(j_{0})$ was the first (and so not the last) requested element of that block 
to evaluate.
If there is no such block $j_{0}$ (if there is, it must be unique), Bob outputs the unique $i'$ for each $0\le j\le 4$ for which Equation (\ref{bobeq}) holds. Either ways the set he outputs has size five.
\end{proof}

Putting Lemmas \ref{product} and \ref{the526lemma} together we get:

\begin{lemma}
There is a $(5^{\ell},30^{\ell})$ strategy for the GKS game for $\ell = 1,2,\ldots$.
\end{lemma}

\begin{corollary}
$k(n) < 5\cdot n^{\log_{30}5}$
\end{corollary}

\section{Further directions}

The strategy in Lemma \ref{the526lemma} generalizes as follows.

\begin{definition}
A $(k,k_{A},n)$ strategy is a usual $(k,n)$ strategy equipped with an additional ``Alice-mode.'' In this Alice evaluates the last position as well (but otherwise she makes the exact same steps as in the usual mode).
Bob does not know if the game is in Alice-mode or not.
We denote the set of outcomes that may arise when the game runs in Alice mode by ${\cal O}_{A}$. The size of the set that Bob sends as an answer to any evaluation of $A$
that is in ${\cal O}_{A}$ has to be at most $k_{A}$.
\end{definition}

\begin{example}
Consider the trivial strategy, where Alice evaluates all (non-final) entries of $A$ to 0. In Alice-mode she also evaluates the final entry, to which she 
always gives the value 1. Bob's strategy is the following: his output 
set includes all $n$ entries of $A$ when he sees only zeroes, but when he sees a 1, his set contains only that entry (he cannot see two ones regardless whether the game has run in Alice-mode or not).
It is easy to see that this is an $(n,1,n)$ strategy.
\end{example}

The above refined parametrization of strategies are useful for the following reason:

\begin{lemma}\label{threetwo}
Assume there is a $(b,a,n)$ strategy. Then there is also a $(b, \lfloor {b\over a}\rfloor \times n)$ strategy.
\end{lemma}
\begin{proof}
Assume that ${\cal O}_{A}$ is set of possible final settings of the array in Alice-mode.
Alice plays the $(b,a,n)$ strategy on $\lfloor {b\over a}\rfloor$ disjoint blocks, each of size $n$, independently, in such a way that in all blocks except in the final one she
plays in Alice-mode. Bob looks at the blocks individually. If there is a block where the final setting is not in ${\cal O}_{A}$, Bob outputs only 
the output set associated with that block. Otherwise he outputs the union of the output sets for all blocks. In the former case the 
output size is at most $b$, in the latter it is at most $\lfloor {b\over a}\rfloor \times a \le b$.
\end{proof}

The above lemma and a $(5,1,6)$ strategy (described in the previous page) yields the proof of Lemma \ref{the526lemma}. In general, it is worthwhile to design $(k,1,n)$ strategies even when $k$ is close to $n$.
We suspect that there is a $(11,1,15)$ strategy. In this Alice exploits a 1-1 correspondence $\tau$ between the set of all four-element subsets of $\{1,\ldots,15\}$ and 
an appropriate ${15\choose 4} = 1365$ size subset of the 1-error correcting 
Hamming code with length 15 and dimension 11. What is missing (but is very likely true) is that $\tau$ can be designed so that for every $H\subseteq [15]$, $|H|=4$, and
for $x = \tau(H)\in \{0,1\}^{15}$ it holds that $x_{H} = 1111$, i.e. $x$ restricted to any coordinate in $H$ is 1. This type of construction goes back to \cite{saks}. 
The strategy: Alice's first four answers 
are always 1, and her other answers make the final vector (when in Alice-mode) to be $\tau(\{i_{1},i_{2},i_{3},i_{4}\})$, where $i_{1},\ldots,i_{4}$ are the first four entries that Merlin requests Alice to evaluate.
If a $(11,1,15)$ strategy exists, by Lemma \ref{threetwo} it should give a $(11,165)$ strategy and by Lemma \ref{product} an $O(n^{0.47})$ upper bound on the GKS game.

\medskip

Even better parameters could result from the 1-error correcting, non-linear code of length 9 that has 40 code words (see \cite{codes}). Here a map, similar to $\tau$ would yield a
$(7,1,9)$ strategy.  We remark that the problem of finding $\tau$ is a bipartite matching problem, thus it can potentially be done with a computer, but in the case of the Hamming code 
there are also more direct ways to try.

\begin{problem}
For given $n>0$ compute or estimate the minimal $k$ such that a $(k,1,n)$ strategy exists.
\end{problem}

\begin{problem}
Is there a $(O(\sqrt{n}),n)$ strategy in which Alice always answers with zero in the first $n- O(\sqrt{n})$ steps?
\end{problem}

\section{Exact bounds for dimensions up to 9}

The  GKS game also has a more combinatorial form \cite{saks}:

\begin{lemma}
There is a $(k,n)$ strategy for the GKS game if and only if there is a subgraph $G$ with maximum degree at most $k$, of the $n$ dimensional hypercube, $K_{2}^{\square n}$, such that 
in the Game below, the Chooser has a winning strategy. 

{\bf Game \cite{saks}:} The game is played by the Divider and the Chooser on $\{0,1\}^{n}$. At each step the Divider picks a coordinate $i\in [n]$ not picked before
and the Chooser decides whether to delete nodes $x$ with $x_{i} = 0$, or with $x_{i} = 1$ from the current set of nodes.
The game lasts $n-1$ rounds and the Chooser wins if the remaining two vertices form an edge of $G$. 
\end{lemma}

Let us call a subgraph of $K_{2}^{\square n}$ as above with maximum degree $k$ a $(k,n)$ subgraph.

\begin{lemma}\label{structure} The structure of the best strategies up to $n=4$ is characterized by:
\begin{enumerate}
\item Any $(1,1)$ subgraph is an edge.
\item Any $(2,2)$ subgraph must contain a subgraph  that is by an automorphism of the square equivalent to $\{\mbox{\rm *0, 0*}\}$ (Figure \ref{case2}).
\item Any $(2,3)$ subgraph must contain a subgraph that is by an automorphism of the cube equivalent to $\{\mbox{\rm 10*, *00, 0*0, 01*}\}$. This is a (particular) path with four edges  (Figure \ref{case34}).
\item Any $(2,4)$ subgraph must contain a subgraph that is by an automorphism of the 4-hypercube equivalent to 
\[
\{
\mbox{\rm 
*000, 
0*00, 
01*0, 
011*, 
0*11, 
*011, 
101*, 
10*0 }
\}
\] 
This is a (particular) cycle of length $8$ (Figure \ref{case34}).
\end{enumerate}
\end{lemma}

\begin{figure}[t]
\begin{center}
   \begin{tikzpicture}
\matrix (m) [matrix of math nodes,
row sep=2em, column sep=2em,
text height=1.5ex,
text depth=0.25ex]{
 01 & 11 \\
00 & 10 \\
};
\path[-]
(m-1-1) edge   (m-1-2) edge  [line width=5pt]  (m-2-1)
(m-2-2) edge (m-1-2)  edge  [line width=5pt]  (m-2-1);
\end{tikzpicture} 
\end{center}
\caption{The minimum $(2,2)$ subgraph. \label{case2}}
\end{figure}
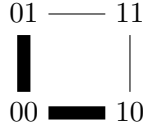

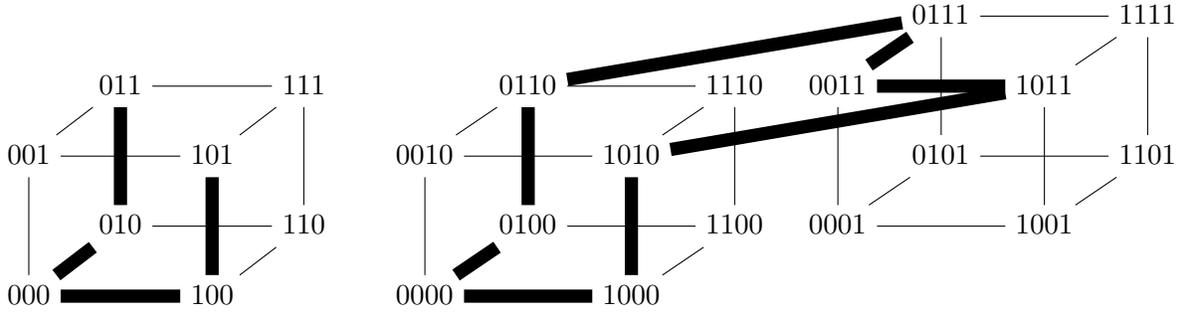
\begin{figure}[t]

\begin{tabular}{ll}
   \begin{tikzpicture}[
]
\matrix (m) [matrix of math nodes,
row sep=1em, column sep=1em,
text height=1.5ex,
text depth=0.25ex]{
& 011  & & 111 \\
001 & & 101 \\
& 010 & & 110 \\
000 & & 100 \\
};
\path[-]
(m-1-2) edge (m-1-4) edge (m-2-1) edge  [line width=5pt] (m-3-2)
(m-1-4) edge (m-3-4) edge (m-2-3)
(m-2-1) edge  (m-2-3) edge (m-4-1)
(m-3-2) edge  (m-3-4) edge  [line width=5pt] (m-4-1)
(m-4-1) edge  [line width=5pt] (m-4-3)
(m-3-4) edge   (m-4-3)
(m-2-3) edge [line width=5pt] (m-4-3);
\end{tikzpicture}  &  
   \begin{tikzpicture}[
]
\matrix (m) [matrix of math nodes,
row sep=1em, column sep=0.9em,
text height=1.5ex,
text depth=0.25ex]{
 & & & & & 0111  & & 1111 \\
& 0110  & & 1110 & 0011 & & 1011 \\
0010 & & 1010 &  & & 0101 & & 1101 \\
& 0100 & & 1100  &  0001 & & 1001 \\
0000 & & 1000 & &  \\
};
\path[-]
(m-1-6) edge  [line width=5pt]  (m-2-2) edge (m-1-8) edge   [line width=5pt]  (m-2-5) edge (m-3-6)
(m-2-5) edge [line width=5pt]  (m-2-7) edge (m-4-5)
(m-2-7) edge [line width=5pt]  (m-3-3) edge (m-1-8)  edge (m-4-7)
(m-4-5) edge (m-3-6) edge (m-4-7)
(m-3-8) edge (m-3-6) edge (m-1-8) edge (m-4-7)
(m-2-2) edge  (m-2-4) edge   [line width=5pt] (m-4-2) edge (m-3-1) 
(m-2-4) edge (m-4-4) edge (m-3-3)
(m-3-1) edge  (m-3-3) edge (m-5-1)
(m-4-2) edge  (m-4-4)
edge  [line width=5pt] (m-5-1)
(m-5-1) edge  [line width=5pt] (m-5-3)
(m-4-4) edge   (m-5-3)
(m-3-3) edge [line width=5pt] (m-5-3);
\end{tikzpicture} \\
\end{tabular}

%
\caption{The minimum $(2,3)$ and $(2,4)$ subgraphs. \label{case34}}
\end{figure}

\begin{proof}
We give a sketch of the proof. Items 1.-3. are easy. Proving 4. seems to require a case separation.
Consider a minimum $(2,4)$ subgraph $G$ (which must be a vertex-disjoint union of cycles and paths).
After the Chooser leaves out half of the vertices, the graph reduces to a $(2,3)$ subgraph, so by item 3, without loss of generality we can assume
that $G$ contains $\{\mbox{\rm 10*0, *000, 0*00, 01*0}\}$. We show that the only 
extension of this edge set to any minimal $(2,4)$ subgraph is the one drawn in Figure \ref{case34}. 

\medskip

\begin{figure}[t]
\begin{center}
   \begin{tikzpicture}[
]
\matrix (m) [matrix of math nodes,
row sep=1em, column sep=0.9em,
text height=1.5ex,
text depth=0.25ex]{
 & & & & & 0111  & & 1111 \\
& 0110  & & 1110 & 0011 & & 1011 \\
0010 & & 1010 &  & & 0101 & & 1101 \\
& 0100 & & 1100  &  0001 & & 1001 \\
0000 & & 1000 & &  \\
};
\path[-]
(m-1-6) edge (m-1-8)  edge (m-3-6) edge  (m-2-5)
(m-2-5) edge (m-2-7) edge (m-4-5)
(m-2-7) edge (m-1-8)  edge (m-4-7)
(m-4-5) edge (m-3-6) edge (m-4-7)
(m-3-8) edge (m-3-6) edge (m-1-8) edge (m-4-7)
(m-2-2) edge  (m-2-4) edge   [line width=5pt] (m-4-2) edge  [line width=5pt]   (m-3-1) 
(m-2-4) edge (m-4-4) edge (m-3-3)
(m-3-1) edge  (m-3-3) edge (m-5-1) edge  [line width=5pt]   (m-2-5)
(m-4-2) edge  (m-4-4)
edge  [line width=5pt] (m-5-1)
(m-5-1) edge  [line width=5pt] (m-5-3)
(m-4-4) edge   (m-5-3)
(m-3-3) edge [line width=5pt] (m-5-3)  ;
\end{tikzpicture} 
\end{center}
\caption{A six-edge configuration that is ruled out. \label{ruledout}}
\end{figure}
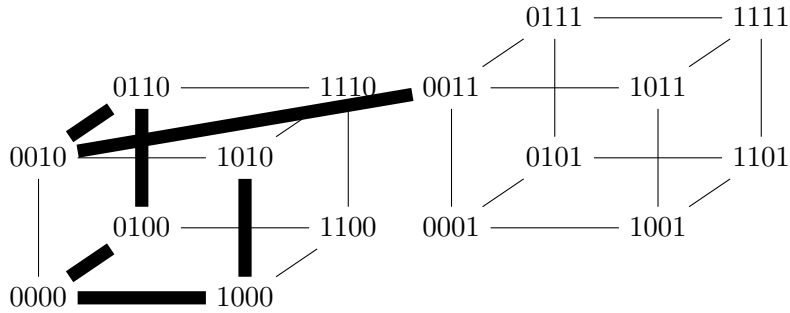

Towards this goal
we further exploit 
the structure that item 3. gives. For $1\le i\le 4$, if we delete all edges in direction $i$ from $G$,
the remaining graph must still contain a path $P_{i}$ of the shape described in 3. By the minimality of $G$ we have $G = P_{1}\cup P_{2}\cup P_{3}\cup P_{4}$.
These $P_{i}$s may share nodes and/or edges. 
We first rule out:

\medskip

{\bf Case.} Node 0010 occurs in the middle of one of $P_{i}$ ($1\le i\le 3$).

\medskip

In this case $G$ must either contain the subgraph in Figure \ref{ruledout}
or its mirror-symmetric version where the 0*10 edge is replaced with the *010 edge. Without loss of generality let us assume the first. Then
the path that contains 0010 in the middle can be either
\begin{eqnarray*}
(A) \;\;\;\;\;\;\;\;\;\; P_{1} &  =  & \{\mbox{01*0, 0*10, 001*, 00*1}\} \;\;\;\; {\rm or} \\
(B) \;\;\;\;\;\;\;\;\;\; P_{1} &  =  & \{\mbox{0*10, 001*, 00*1, 0*01}\} \;\;\;\; {\rm or} \\
(C) \;\;\;\;\;\;\;\;\;\; P_{3} &  =  & \{\mbox{01*0, 0*10, 001*, 00*1}\}  \\
\end{eqnarray*}
We first show the impossibility of case (A), i.e. if $G$ contains the pattern:

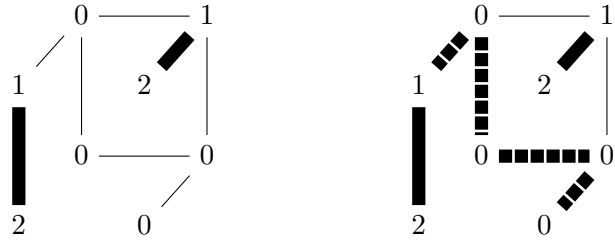
\begin{center}
   \begin{tikzpicture}
\matrix (m) [matrix of math nodes,
row sep=1em, column sep=0.9em,
text height=1.5ex,
text depth=0.25ex]{
 & & & & & 0111  & & 1111 \\
& 0110  & & 1110 & 0011 & & 1011 \\
0010 & & 1010 &  & & 0101 & & 1101 \\
& 0100 & & 1100  &  0001 & & 1001 \\
0000 & & 1000 & &  \\
};
\path[-]
(m-1-6) edge (m-1-8)  edge (m-3-6) edge  (m-2-5)
(m-2-5) edge (m-2-7) edge [line width=5pt, style={densely dotted}]  (m-4-5)
(m-2-7) edge (m-1-8)  edge (m-4-7)
(m-4-5) edge (m-3-6) edge (m-4-7)
(m-3-8) edge (m-3-6) edge (m-1-8) edge (m-4-7)
(m-2-2) edge  (m-2-4) edge   [line width=5pt] (m-4-2) edge  [line width=5pt]   (m-3-1) 
(m-2-4) edge (m-4-4) edge (m-3-3)
(m-3-1) edge  (m-3-3) edge (m-5-1) edge  [line width=5pt]   (m-2-5)
(m-4-2) edge  (m-4-4)
edge  [line width=5pt] (m-5-1)
(m-5-1) edge  [line width=5pt] (m-5-3)
(m-4-4) edge   (m-5-3)
(m-3-3) edge [line width=5pt] (m-5-3)  ;
\end{tikzpicture} 
\end{center}

The undoing of (A) is that we cannot accommodate $P_{3}$. To see this, imagine that we take out all of the vertical edges and look for $P_{3}$ either
in the upper half ($x_{3}=1$) or in the lower half  ($x_{3}=0$). In the upper half we have four free nodes and the node 1010, which has degree one. 
The latter should be a starting node of $P_{3}$ (it is one of the only five nodes  we can use, so $P_{3}$ must use it). 
For adjacent edge only 101* or 1*10 come into question. The other final edge of $P_{3}$ then would be 011* or 0*11, respectively, each resulting in nodes of degree three.
In the lower half we face a similar problem. Since node 1100 must be used in $P_{3}$ as an end point with adjacent edge 110*, the other final edge of $P_{3}$ must be 000*,
creating not even only one, but two nodes with degree three (0000 and 0001).
Case (B) is subsumed by (A), since it gives a strictly larger $P_{4}\cup P_{1}$.

\medskip

Similar argument rules out case (C):

\begin{center}
   \begin{tikzpicture}
\matrix (m) [matrix of math nodes,
row sep=1em, column sep=0.9em,
text height=1.5ex,
text depth=0.25ex]{
 & & & & & 0111  & & 1111 \\
& 0110  & & 1110 & 0011 & & 1011 \\
0010 & & 1010 &  & & 0101 & & 1101 \\
& 0100 & & 1100  &  0001 & & 1001 \\
0000 & & 1000 & &  \\
};
\path[-]
(m-1-6) edge (m-1-8)  edge (m-3-6) edge  (m-2-5)
(m-2-5) edge  [line width=5pt, style={densely dotted}]  (m-2-7) edge (m-4-5)
(m-2-7) edge  [line width=5pt, style={densely dotted}]  (m-1-8)  edge (m-4-7)
(m-4-5) edge (m-3-6) edge (m-4-7)
(m-3-8) edge (m-3-6) edge (m-1-8) edge (m-4-7)
(m-2-2) edge  (m-2-4) edge   [line width=5pt] (m-4-2) edge  [line width=5pt]   (m-3-1) 
(m-2-4) edge (m-4-4) edge (m-3-3)
(m-3-1) edge  (m-3-3) edge (m-5-1) edge  [line width=5pt]   (m-2-5)
(m-4-2) edge  (m-4-4)
edge  [line width=5pt] (m-5-1)
(m-5-1) edge  [line width=5pt] (m-5-3)
(m-4-4) edge   (m-5-3)
(m-3-3) edge [line width=5pt] (m-5-3)  ;
\end{tikzpicture} 
\end{center}

We try to accommodate $P_{1}$. Accommodating $P_{1}$ in the half $x_{1}=0$ is entirely out of question, since there are only three vertices available. The picture below shows
the existing edges and degrees in the $x_{1}=1$ half (the $x_{4}$ direction is changed to horizontal) and the only 
way to accommodate $P_{1}$ (dotted line) in that half:

\begin{center}
\begin{tabular}{lll}
\begin{tikzpicture}
\matrix (m) [matrix of math nodes,
row sep=1em, column sep=1em,
text height=1.5ex,
text depth=0.25ex]{
& 0  & & 1 \\
1 & & 2 \\
& 0 & & 0 \\
2 & & 0 \\
};
\path[-]
(m-1-2) edge (m-1-4) edge (m-2-1) edge (m-3-2)
(m-1-4) edge (m-3-4) edge  [line width=5pt]  (m-2-3)
(m-2-1)  edge [line width=5pt] (m-4-1) 
(m-3-2) edge (m-3-4) 
(m-3-4) edge   (m-4-3);
\end{tikzpicture}  &  \hspace{0.5in} &
\begin{tikzpicture}
\matrix (m) [matrix of math nodes,
row sep=1em, column sep=1em,
text height=1.5ex,
text depth=0.25ex]{
& 0  & & 1 \\
1 & & 2 \\
& 0 & & 0 \\
2 & & 0 \\
};
\path[-]
(m-1-2) edge (m-1-4) edge [line width=5pt, style={densely dotted}]  (m-2-1) edge [line width=5pt, style={densely dotted}]  (m-3-2)
(m-1-4) edge (m-3-4) edge  [line width=5pt]  (m-2-3)
(m-2-1)  edge [line width=5pt] (m-4-1) 
(m-3-2) edge [line width=5pt, style={densely dotted}]  (m-3-4) 
(m-3-4) edge  [line width=5pt, style={densely dotted}]  (m-4-3);
\end{tikzpicture} 
\end{tabular}
\end{center}

In the large picture:

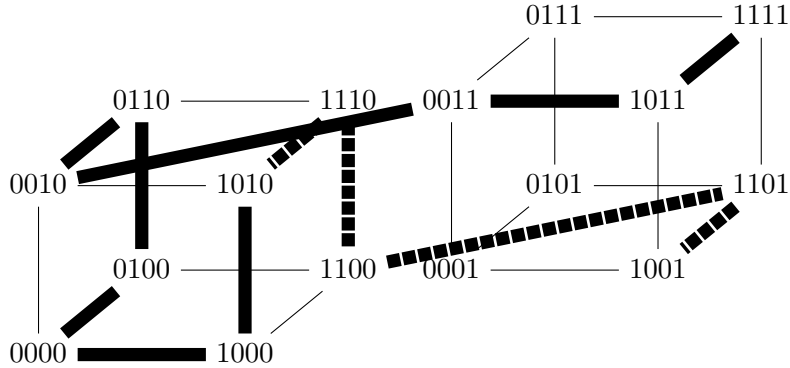
\begin{center}
   \begin{tikzpicture}
\matrix (m) [matrix of math nodes,
row sep=1.5em, column sep=0.9em,
text height=1.5ex,
text depth=0.25ex]{
 & & & & & 0111  & & 1111 \\
& 0110  & & 1110 & 0011 & & 1011 \\
0010 & & 1010 &  & & 0101 & & 1101 \\
& 0100 & & 1100  &  0001 & & 1001 \\
0000 & & 1000 & &  \\
};
\path[-]
(m-1-6) edge (m-1-8)  edge (m-3-6) edge  (m-2-5)
(m-2-5) edge  [line width=5pt]  (m-2-7) edge (m-4-5)
(m-2-7) edge  [line width=5pt]  (m-1-8)  edge (m-4-7)
(m-4-5) edge (m-3-6) edge (m-4-7)
(m-3-8) edge (m-3-6) edge (m-1-8) edge [line width=5pt, style={densely dotted}] (m-4-7)
(m-2-2) edge  (m-2-4) edge   [line width=5pt] (m-4-2) edge  [line width=5pt]   (m-3-1) 
(m-2-4) edge [line width=5pt, style={densely dotted}]  (m-4-4) edge [line width=5pt, style={densely dotted}]  (m-3-3)
(m-3-1) edge  (m-3-3) edge (m-5-1) edge  [line width=5pt]   (m-2-5)
(m-4-2) edge  (m-4-4)
edge  [line width=5pt] (m-5-1)
(m-5-1) edge  [line width=5pt] (m-5-3)
(m-4-4) edge   (m-5-3) edge  [line width=5pt, style={densely dotted}]  (m-3-8)
(m-3-3) edge [line width=5pt] (m-5-3)  ;
\end{tikzpicture} 
\end{center}

But then $P_{2}$ cannot be accommodated. 

\medskip

This does not only finish the impossibility of the {\bf Case}, but it gives the general statement that the pattern

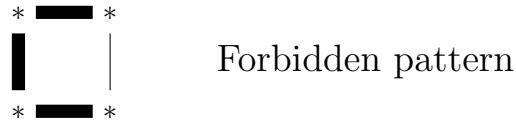
\begin{center}
   \begin{tikzpicture}
\matrix (m) [matrix of math nodes,
row sep=2em, column sep=2em,
text height=1.5ex,
text depth=0.25ex]{
 * & * \\
* & * \\
};
\path[-]
(m-1-1) edge  [line width=5pt]  (m-1-2) edge  [line width=5pt]  (m-2-1)
(m-2-2) edge (m-1-2)  edge  [line width=5pt]  (m-2-1);
\node[draw=white] at (4,0) {\Large Forbidden pattern}; 
\end{tikzpicture} 
\end{center}

should not occur in $G$. The reason is that the above pattern {\em in a minimal $G$} occurs if and only if the shape in Figure \ref{ruledout} occurs in $G$. The argument supporting this goes that
in a minimal $G$ the above pattern must be covered by 
a union of two $P_{i}$s. The one that contains two of the highlighted edges (at least one of the $P_{i}$s must be such) without loss of generality can be identified $P_{4}$ where the 
identification also has the property that the two 
edges in question are  0*00 and  01*0.
The other $P_{i'}$ must then contain 0010 as a middle point, referring us to {\bf Case}.

\medskip

We now develop a new representation for the shape of a path: we go through the path from one end-edge to the other and list the directions in which the edges go. 
In the case of a cycle we pick an arbitrary starting edge. When the cycle is a connected component of $G$ (this is the case we are interested in), we can always start at the start of a $P_{i}$.
By further exploiting symmetries we set the start (whether a path or a cycle) to 3123 (of curse, only when we are searching for shapes of {\em single} connected components of $G$).
In our new representation for instance the path and cycle in Figure \ref{case34} have shapes 3123 and 31234214, respectively.
We now build a tree representing the shapes of all potential path- or cycle- components of $G$ (up to isomorphism) taking into 
consideration that a.) the path must be simple or a cycle b.) the forbidden pattern "aba" should not occur, and c.) 
the path (or cycle) must be a union of some $P_{i}$s:

\medskip

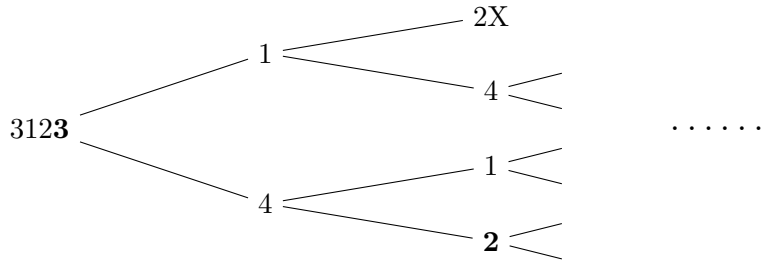
\begin{center}

\tikzstyle{level 1}=[level distance=3.5cm, sibling distance=3.5cm]
\tikzstyle{level 2}=[level distance=3.5cm, sibling distance=1.5cm]
\tikzstyle{level 3}=[level distance=0.5cm, sibling distance=0.1cm]

\tikzstyle{bag} = [text width=4em, text centered]
\tikzstyle{end} = [rectangle, draw=none, minimum width=3pt, inner sep=0pt]

\tikzstyle{ans} = [color=red]

\begin{tikzpicture}[level distance=3cm,
level 1/.style={sibling distance=2cm},
level 2/.style={sibling distance=1cm},
level 3/.style={sibling distance=0.5cm, level distance = 1cm},grow'=right]
\tikzstyle{every node}=[]
    \node (Root) [] {312{\bf 3}}
        child [] {
        node {1}
        child { node {2X} 
                edge from parent
                node[left] {}
        }
              child [black] { node {4} 
                child {node[ans,end] {} }
                 child {node[ans,end] {} }
            edge from parent
            node[ans,left] {}
        }
        edge from parent
    }
 child {
        node {4}
        child { node {1} 
                child {node[ans,end] {} }
                   child {node[ans,end] {} }
                edge from parent
                node[ans,left] {}
        }
              child { node {{\bf 2}}
                child {node[end] {} }
                   child {node[ans,end] {} }
                edge from parent
                node[ans,left] {}
        }
        edge from parent
    };
    \node[draw=white] at (9,0) {\Large $\ldots\ldots$}; 
\end{tikzpicture}

\end{center}

\medskip

We get that any path or cycle that {\em starts with} 3123 and which can occur in $G$ must be one of:

\medskip

\begin{center}
\begin{tabular}{llll}
Type 1: & 3123 & contains $P_{4}$; & 5 nodes path\\
Type 2: & 31231421 & contains $P_{4}$ and $P_{3}$; & 9 nodes path \\
Type 3: & 3123143 & contains $P_{4}$ and $P_{2}$; & 8 nodes path \\
Type 4: & 31234124 & contains $P_{4}$ and $P_{3}$; & 8 nodes cycle \\
Type 5: & 3123413 & contains $P_{4}$ and $P_{2}$; & 8 nodes path \\
Type 6: & 312342 & contains $P_{4}$ and $P_{1}$; & 7 nodes path \\
Type 7: & 31234214 & contains $P_{1}$, $P_{2}$, $P_{3}$, $P_{4}$; & 8 nodes cycle \\
\end{tabular}
\end{center}

\medskip

The shapes of all components of $G$ must be {\em equivalent} to one of these types, where on equivalence we mean that we allow to permute $\{1,2,3,4\}$. From 
now on, on a shape we mean the entire equivalence class.
A cycle of Type 7 satisfies the criteria for $G$, and our goal is to show that there is no other solution. 
To look for a further (minimal) solution we can omit Type 7 from our list.
By counting nodes we obtain that the connected components of $G$ cannot be all of Type 1. A single Type 1- component 
and a single other one from types 2-6 contain only three of the $P_{i}$s. On the other hand, a three or more component $G$ 
with any other than Type 1 components would have too many nodes. So we can eliminate Type 1 components altogether from our consideration.
Perhaps the simplest way to proceed from here is just to look at each type from 2 to 6 and check if we can accommodate all $P_{i}$s.
As an example we show how to exclude Type 6, which may seem the most economic of all the types among 2-6:

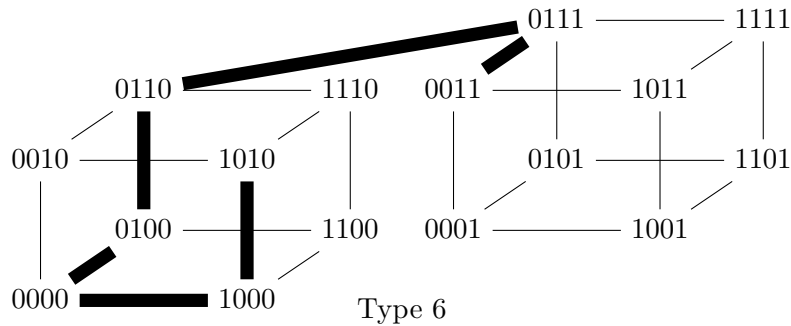
\begin{center}
   \begin{tikzpicture}
\matrix (m) [matrix of math nodes,
row sep=1em, column sep=0.9em,
text height=1.5ex,
text depth=0.25ex]{
 & & & & & 0111  & & 1111 \\
& 0110  & & 1110 & 0011 & & 1011 \\
0010 & & 1010 &  & & 0101 & & 1101 \\
& 0100 & & 1100  &  0001 & & 1001 \\
0000 & & 1000 & &  \\
};
\path[-]
(m-1-6) edge (m-1-8)  edge (m-3-6) edge    [line width=5pt]  (m-2-5) edge  [line width=5pt]   (m-2-2)
(m-2-5) edge  (m-2-7) edge (m-4-5)
(m-2-7) edge  (m-1-8)  edge (m-4-7)
(m-4-5) edge (m-3-6) edge (m-4-7)
(m-3-8) edge (m-3-6) edge (m-1-8) edge (m-4-7)
(m-2-2) edge  (m-2-4) edge   [line width=5pt] (m-4-2) edge   (m-3-1) 
(m-2-4) edge (m-4-4) edge (m-3-3)
(m-3-1) edge  (m-3-3) edge (m-5-1) 
(m-4-2) edge  (m-4-4)
edge  [line width=5pt] (m-5-1)
(m-5-1) edge  [line width=5pt] (m-5-3)
(m-4-4) edge   (m-5-3)
(m-3-3) edge [line width=5pt] (m-5-3)  ;
\node[draw=white] at (0,-2) {Type 6}; 
\end{tikzpicture} 
\end{center}

The solid line is now a whole  {\em connected component} (of Type 6), so $P_{2}$ and $P_{3}$, the remaining missing $P_{i}$s cannot share any vertex with this path.
We argue that it is impossible to accommodate $P_{3}$. The $x_{3}=1$ half has only four free nodes, thus we have to accommodate $P_{3}$ in the 
$x_{3}=0$ half. The latter has 5 free nodes, but the 
starting point of $P_{3}$ has to be 1100 with attached edge 110*. But then $P_{3}$ would also have to contain the edge 000*, which collides with our component.
In fact, the above argument almost without any change works for types 2-5 as well: we cannot accommodate $P_{3}$ on the $x_{3}=1$ and $x_{3}=0$ halves for the very same 
(or in some cases even simpler) reasons. We are done, since we have excluded all types of components except Type 7. 
We can also observe (in order to prove the exact statement promised in the beginning), that there is a unique Type 7 cycle that contains our initial $P_{4}$, the one drawn in
Figure \ref{case34}. \end{proof}

\begin{lemma}
There is no $(2,5)$ strategy.
\end{lemma}

\begin{proof}
Assume on the contrary that there is a $(2,5)$ subgraph $G$ (corresponding to a $(2,5)$ strategy). 
When the Divider picks $x_{i}$, and the Chooser picks a side (either $x_{i}=0$ or $x_{i}=1$), the remaining subgraph must be a $(2,4)$ subgraph, hence it
must contain a cycle $C_{i}$ of length 8 of the 
shape described in Lemma \ref{structure}. Moreover, the cycle $C_{i}$ contains edges in all directions except in direction $i$. 
Because the maximum degree of $G$ is at most two, any two cycles in $G$ must either coincide or must be disjoint.
For the above two reasons for $1\le i\neq j\le 5$ the cycles $C_{i}$ and $C_{j}$ have to be disjoint. 
But $5\times 8 = 40 > 32$, a contradiction.
\end{proof}

Since it is known that there is a $(3,9)$ strategy (in general, $(k,k^2)$ strategy), we have:

\medskip

\begin{center}
\begin{tabular}{clllllllllc}
$n$ & 1 & 2 & 3 & 4 & 5 & 6 & 7 & 8 & 9 & 10 \\
best $k$  & 1 & 2 & 2 & 2 & 3 & 3 & 3 & 3 & 3 & ? \\
\end{tabular}
\end{center}

\end{document}